\documentclass{article}
\usepackage{arxiv}
\usepackage[utf8]{inputenc} 
\usepackage[T1]{fontenc}    
\usepackage{hyperref}       
\usepackage{url}            
\usepackage{booktabs}       
\usepackage{amsmath,amssymb,amsfonts,amsthm} 
\usepackage{nicefrac}       
\usepackage{microtype}      
\usepackage{lipsum}		
\usepackage{graphicx}
\usepackage{natbib}
\usepackage{doi}

\newcommand{\beq}{\begin{equation}}
\newcommand{\eeq}{\end{equation}}

\newtheorem{theorem}{Theorem}

\title{Point Mass in the
Confidence Distribution:
Is it a Drawback or an Advantage?}

\date{}

\author{ 
\href{https://orcid.org/0000-0000-0000-0000}
{\includegraphics[scale=0.06]{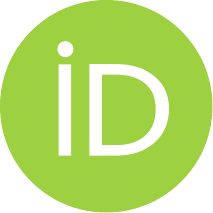}\hspace{1mm}
Hangbin~Lee}\\
Department of Statistics\\
Seoul National University\\
Seoul, 08826 \\
\texttt{hangbin221@gmail.com} \\
\And
\href{https://orcid.org/0000-0002-3447-4306}
{\includegraphics[scale=0.06]{orcid.pdf}\hspace{1mm}
Youngjo~Lee} \\
Department of Statistics\\
Seoul National University\\
Seoul, 08826 \\
\texttt{youngjo@snu.ac.kr} \\
}

\renewcommand{\shorttitle}{Foundations of h-likelihood}

\begin{document}
\maketitle

\begin{abstract}
Stein's \citeyearpar{stein59} problem highlights the phenomenon
called the probability dilution in high dimensional cases,
which is known as a fundamental deficiency in probabilistic inference.
The satellite conjunction problem also suffers from probability dilution
that poor-quality data can lead to a dilution of collision probability. 
Though various methods have been proposed,
such as generalized fiducial distribution
and the reference posterior,
they could not maintain the coverage probability 
of confidence intervals (CIs) in both problems.
On the other hand, the confidence distribution (CD)
has a point mass at zero, which has been interpreted paradoxical.
However, we show that this point mass is an advantage rather than a drawback,
because it gives a way to maintain the coverage probability of CIs.
More recently, `false confidence theorem' was presented
as another deficiency in probabilistic inferences, 
called the false confidence.
It was further claimed that the use of consonant belief 
can mitigate this deficiency.
However, we show that the false confidence theorem
cannot be applied to the CD 
in both Stein's and satellite conjunction problems.
It is crucial that a confidence feature, 
not a consonant one, 
is the key to overcome the deficiencies in probabilistic inferences.
Our findings reveal that the CD outperforms the other existing methods, 
including the consonant belief,
in the context of Stein's and satellite conjunction problems. 
Additionally, we demonstrate the ambiguity of coverage probability 
in an observed CI from the frequentist CI procedure, 
and show that the CD provides valuable information 
regarding this ambiguity.
\end{abstract}


\section{Introduction}

\label{s: introduction} 
In 1930, \citeauthor{fisher30} introduced the fiducial distribution (FD) 
as an alternative to the Bayesian posterior distribution,
without requiring a prior.
The concept of confidence can be understood as 
a Neymanian interpretation of Fisher's FD \citep{schweder16}.
However, Fisher objected Neyman's notion of coverage probability,
as he believed that scientists would never repeat the same experiment.
Fisher's FD has encountered substantial criticism over the century,
with ongoing controversy about 
its adherence to Kolmogorov's axioms \citep{edwards77}.
\citet{wilkinson77} argued that the FD is not a probability distribution
and criticized it for its non-coherency.
Despite the controversies surrounding the FD,
there has been a growing interest, 
especially as the confidence distribution (CD),
which offers a more refined definition.
As \citet{efron98} noted, 
``Maybe Fisher's biggest blunder will become a big hit in the 21st century!'',
it is important to understand the advantages of CD. 
In this paper, we collectively refer to FD as CD, 
in accordance with the suggestion of \citet{efron98}.
We demonstrate the complementary nature of 
epistemic confidence and frequentist coverage probability,
emphasizing their synergy rather than conflict.

In Bayesian analysis, integrating out the nuisance parameter is 
often legitimate to obtain the marginal posterior. 
However, analogous to the marginalization paradox \citep{dawid73},
integrating a CD may no longer result in a valid CD.
Stein's \citeyearpar{stein59} problem serves as a benchmark 
for illustrating this phenomenon of the integrated CDs,
which leads to the CIs exhibiting undesirable behavior 
\citep{wilkinson77, schweder16}.
Here, the integrated CD is the same with 
the marginal posterior under uniform prior (UP).
To resolve this paradoxical behavior,
\citet{wilkinson77} proposed the use of the marginally defined CD, 
and \citet{bernardo79} introduced the concept of the reference prior (RP).
\citet{pawitan21} showed that the CD is not a probability 
but an extended likelihood, which is necessary to avoid 
probability-related paradoxes \citep{pawitan17}. 
Clearly, integrating out the nuisance parameters is not a legitimate way 
to obtain the marginal extended likelihood.
This implies that the integrated CD is not a proper CD in general,
and it is no longer a paradox because the CD is not a probability.
The extended likelihood principle \citep{bjornstad96} states that 
the extended likelihood contains all the information in the data.
\citet{pawitan23} proposed an epistemic CD associated with 
the full data likelihood, which leaves no room for a relevant subset.
They obtained a marginal CD with the full data likelihood 
by conditioning on maximal ancillary statistics.

In this paper, we show that the CD does not suffer from
a fundamental deficiency of probabilistic inference,
called the probability dilution.
In satellite conjunction problem, 
the probability of collision plays an important role in risk assessment. 
However, a probability dilution occurs in which 
the use of lower-quality data appears to reduce the probability of collision. 
This results in a severe and persistent underestimate of risk exposure. 
\citet{cunen20} derived a marginal CD for satellite conjunction problem, 
but \citet{martin21} pointed out that 
the marginal CD has a drawback of not allowing two-sided CIs. 
Stein's problem has a probability dilution 
with increasing number of parameters $k$, 
whereas the satellite problem has a probability dilution 
with increasing variance $\sigma^2$. 
\citet{hannig09} generalized the FD by the GFD,
which can be easily implemented in practical applications, 
but he claimed not to attempt to derive a paradox-free FD. 
We show that the integrated CD in Stein's problem can be viewed as a GFD. 
This implies that the GFD may not avoid the probability dilution.
Throughout the paper, we use UP and RP to denote 
both the priors and their corresponding marginal posteriors, 
depending on the context.
The RP can improve the UP with moderate $k$ and $\sigma$, 
but it cannot avoid the probability dilution entirely. 
In this paper, we study the role of point mass in the CD. 
The point mass in CD has been considered as paradoxical 
\citep{wilkinson77, schweder16}, 
but we show that the point mass is necessary 
to maintain the stated coverage probability of CI (confidence feature).
In order to maintain such a confidence feature,
occasionally the CD may not allow two-sided observed CI.
Thus, this property is not a drawback but an advantage of the CD.
In consequence, the CD can overcome the probability dilution entirely,
whereas the UP and RP, which do not have a point mass, 
cannot maintain the confidence feature.

Recently, \citep{balch19} raised another deficiency in probabilistic inference,
called the false confidence,
and introduced the Martin-Liu validity criterion \citep{martin15}
to identify statistical methods that are free from false confidence.
An additional consonant feature \citep{balch12} is claimed to be 
necessary for the CD to overcome the false confidence
\citep{denoeux18, balch20, martin21}.
In particular, \citet{martin21} claimed that 
the false confidence theorem \citep{balch19} can be applied to
not only the Bayesian posteriors
but also the CD in satellite conjunction problem,
which implies that the CD is at risk of false confidence.
However, 
we demonstrate that the false confidence theorem \citep{balch19}
cannot be applied to the CD in both Stein's and satellite conjunction analysis.
We further introduce the null belief theorem,
which could be interpreted as a side effect of 
satisfying the Martin-Liu validity criterion for any possible false propositions.
This precludes the use of an additional consonant feature
and raise a question; is the consonant feature indeed necessary for the CD?

In satellite conjunction problem, risk assessment of collision is 
often based on binary hypothesis testing procedures \citep{hejduk19}.
Due to the null belief, the CB cannot be used for the testing procedure
and performs uniformly worse than the confidence-based decision making. 
UP and RP also produce small collision probability under poor data 
to make a wrong decision making on satellite collision. 
For the CD to become a big hit, 
it should have an advantage over the existing methods. 
In Stein's and the satellite conjunction problems,
the presence of a point mass allows the CD 
to avoid both probability dilution and false confidence
for the proposition of interest.

\section{Ambiguity in Confidence Level of an observed CI}

\label{s: freq}

Stein's \citeyearpar{stein59}
problem plays an important role 
in illustrating the phenomenon of probability dilution
in high-dimensional cases. 
Suppose that $Y_{1},...,Y_{k}$ are independent random variables with 
$Y_{i}\sim N(\mu_{i},\sigma^{2})$ for $i=1,...,k$.
The parameter of interest is 
$\theta =||\mu ||=\sqrt{\sum_{i=1}^{k}\mu_{i}^{2}}$. 
When $\theta =0$,
the statistic $\sum_{i=1}^{k}Y_{i}^{2}$ follows 
a central Chi-squared distribution, and when $\theta >0$, 
it follows a non-central Chi-squared distribution.
\citet{stein59} demonstrated the probability dilution
in the marginal posterior under the UP when $\theta \ll k$.
More recently, there has been a growing interest
in the satellite conjunction problem, 
which can be simplified as a two-dimensional problem. 
Following \citet{cunen20} and \citet{martin21}, 
let $(Y_{1},Y_{2})$ represent the measurements 
of the true difference between the positions of two satellites,
denoted as $(\mu_{1},\mu_{2})$ along each axis. 
Suppose that $Y_{1}\sim N(\mu_{1},\sigma^{2})$ and 
$Y_{2}\sim N(\mu_{2},\sigma^{2})$ are independent.
The parameter of interest is the Euclidean distance 
between the two satellites, denoted as
$\theta =\sqrt{\mu_{1}^{2}+\mu_{2}^{2}}$. 
The measurement of $\theta$ is represented by 
$D=\sqrt{Y_{1}^{2}+Y_{2}^{2}}$. 
Thus, the parameter space $\Theta$
and sample space $\Omega_D$ are identical,
$\Theta = \Omega_D = [0, \infty)$.
To distinguish the random variables $Y_{1},Y_{2}$ and $D$ 
from their observed values, 
the latter are written in lower cases $y_{1},y_{2}$ and $d$. 
It is worth emphasizing that in the satellite problem,
the probability dilution becomes apparent as $\sigma \to \infty$,
while in Stein's problem,
it becomes severe as $k \to \infty$.
For the simplicity in our discussions,
we mainly use the satellite problem with $k=2$ for derivations.
However, our results can be applied to Stein's problem as well.
In both Stein's and satellite conjunction problems,
$\sigma$ is assumed to be known.

We start with demonstrating that
an observed CI can exhibit ambiguity
in its confidence level (coverage probability).
Suppose that we want to make a frequentist CI procedure 
based on the statistics $D$, having 
\begin{equation}
\frac{D^{2}}{\sigma^{2}}\sim \chi _{2}^{2}\left( \frac{\theta^{2}}{\sigma
^{2}}\right) ,  \label{eq:d}
\end{equation}
where $\chi _{df}^{2}(\nu )$ denotes the non-central Chi-square distribution
with the degrees of freedom $df$ and non-centrality parameter $\nu $.
Suppose that $q_{\alpha}(\theta )$ is the $(1-\alpha )$ quantile function
of $D$ such that 
\[
P_{\theta}(D\leq q_{\alpha}(\theta ))=1-\alpha , 
\]
where $\alpha \in (0,1)$. 
Since $q_{\alpha}(\theta )$ is a strictly increasing function of $\theta$ 
for any $\alpha \in (0,1)$, 
there exists an inverse function $q_{\alpha}^{-1}(d)$ such that 
\[
P_{\theta =q_{\alpha}^{-1}(d)}(D\leq d)
=P_{\theta}(q_{\alpha}^{-1}(D)\leq\theta)
=1-\alpha. 
\]
However, the range of $q_{\alpha}(\theta )$ is $[q_{\alpha}(0),\infty)$
where $q_{\alpha}(0)$ is the $(1-\alpha )$ quantile of 
the central Chi-square distribution. 
Therefore, $q_{\alpha}^{-1}(d)$ is not defined for $d<q_{\alpha}(0)$. 
We let $q_{\alpha}^{-1}(d)=0$ for such $d$. 
For $\alpha=0$ and $1$, 
let $q_{0}^{-1}(d)=\lim_{\alpha \rightarrow 0}q_{\alpha}^{-1}(d)=0$ 
and $q_{1}^{-1}(d)=\lim_{\alpha \rightarrow 1}q_{\alpha}^{-1}(d)=\infty$.
To have a frequentist CI procedure with $\alpha$-level of confidence,
in this paper we consider the following CI procedures,
\begin{equation}
\text{CI}_{\alpha}(D)=[\theta_{L}(D),\theta_{U}(D)),  \label{eq:CI}
\end{equation}
and 
\begin{equation}
\text{CI}^*_{\alpha}(D)=[\theta_{L}(D),\theta_{U}(D)],  \label{eq:CI_closed}
\end{equation}
where $\theta_{L}(D)=q_{1-\alpha -\beta}^{-1}(D)$ 
and $\theta_{U}(D)=q_{1-\beta}^{-1}(D)$ 
for some $0\leq \beta \leq 1-\alpha$. 
$\text{CI}_\alpha^*(D)$ is a closed interval,
whereas $\text{CI}_\alpha(D)$ is a half-closed interval.
When $\beta=0$ or $1-\alpha$, the CI procedure become one-sided,
either $[\theta_{L}(D),\infty )$ or $[0,\theta_{U}(D))$, respectively. 
When $0<\beta <1-\alpha$, it becomes two-sided.
For example, if $\alpha =0.9$ and $\beta =0.05$, 
$\theta_{L}(D)=q_{0.05}^{-1}(D)$ and $\theta_{U}(D)=q_{0.95}^{-1}(D))$ 
with the confidence level $\alpha =0.95-0.05=0.9$. 
Then, the coverage probabilities of the CI procedures
\eqref{eq:CI} and \eqref{eq:CI_closed}
are equivalent to the confidence level $\alpha$ for all $\theta>0$,
\begin{align*}
P_{\theta}(\theta \in \text{CI}_{\alpha}(D))
&=P_{\theta}(\theta_{L}(D)\leq \theta < \theta_{U}(D))
=P_{\theta}(\theta_{L}(D)\leq \theta)-P_{\theta}(\theta_{U}(D)\leq \theta)
=\alpha +\beta -\beta
=\alpha,
\\
P_{\theta}(\theta \in \text{CI}^*_{\alpha}(D))
&=P_{\theta}(\theta_{L}(D)\leq \theta \leq \theta_{U}(D))
=P_{\theta}(\theta_{L}(D)\leq \theta)-P_{\theta}(\theta_{U}(D) < \theta)
=\alpha + P_{\theta}(\theta_{U}(D) = \theta)
=\alpha. 
\end{align*}
However, for $\theta=0$, the CI procedure \eqref{eq:CI} has
$$
P_{\theta=0}(0\in \text{CI}_{\alpha}(D))
=P_{\theta=0}(\theta_{L}(D)=0<\theta_{U}(D))
=\alpha,
$$
whereas the CI procedure \eqref{eq:CI_closed} has
$$
P_{\theta=0}(0\in \text{CI}^*_{\alpha}(D))
=P_{\theta=0}(\theta_{L}(D)=0\leq\theta_{U}(D))
=\alpha + P_{\theta=0}(\theta_{L}(D)=\theta_{U}(D)=0)
>\alpha.
$$
Given the observed data $D=d$, 
the two-sided CI procedure \eqref{eq:CI} 
with $0<\beta<1-\alpha$ 
leads to the following observed CI,
$\text{CI}_{\alpha}(d)=[\theta_{L}(d),\theta_{U}(d))$:
\begin{itemize}
\item[(a)]  If $d>q_{1-\alpha -\beta}(0)$, the observed CI becomes
two-sided interval $[\theta_{L}(d),\theta_{U}(d))$
with $\theta_L(d) > 0$.
\item[(b)]  If $q_{1-\beta}(0)<d\leq q_{1-\alpha -\beta}(0)$, the observed
CI becomes one-sided interval $[0,\theta_{U}(d))$.
\item[(c)]  If $d\leq q_{1-\beta}(0)$, the observed CI becomes an empty
interval $[0,0)$.
\end{itemize}
On the other hand, the CI procedure \eqref{eq:CI_closed}
leads to the following observed CI:
\begin{itemize}
\item[(a)]  If $d>q_{1-\alpha -\beta}(0)$, the observed CI becomes
two-sided closed interval $[\theta_{L}(d),\theta_{U}(d)]$.
\item[(b)]  If $q_{1-\beta}(0)<d\leq q_{1-\alpha -\beta}(0)$, the observed
CI becomes one-sided closed interval $[0,\theta_{U}(d)]$.
\item[(c)]  If $d\leq q_{1-\beta}(0)$, the observed CI becomes $\{0\}$.
\end{itemize}

\noindent Thus, the main difference between the observed CIs from 
\eqref{eq:CI} and \eqref{eq:CI_closed} occurs in case (c);
\eqref{eq:CI} provides an empty set 
and \eqref{eq:CI_closed} provides an interval $\{0\}$.
Although the empty CI seems not natural,
it is important to note that only the procedure \eqref{eq:CI}
maintains the coverage probability $\alpha$ 
for all $\theta \in \Theta$, including $\theta=0$,
whereas the procedure \eqref{eq:CI_closed} gives 
a conservative interval at $\theta=0$.
However, the open CI procedure $(\theta_L(D), \theta_U(D))$
cannot maintain the coverage probability.

As an illustrative example, 
consider the case where $\sigma=1$ and $\beta =0.05$.
Figure \ref{fig:ci} presents three CI procedures \eqref{eq:CI}
with different confidence level;
one-sided CI procedure with $\theta_{L}(D)=0$ for $\alpha=0.95$ 
and two-sided CI procedures for $\alpha=0.90$ and $0.60$.
All three CIs have a common upper bound,
$\theta_{U}(d)=q_{1-\beta}^{-1}(d)=q_{0.95}^{-1}(d)$. 
In the figures, the horizontal axis and vertical axis represent 
$d$ and $\theta$, respectively. 
Dashed lines and the solid lines represent
$\theta_{U}(d)$ and $\theta_{L}(d)$, respectively. 
For $\alpha=0.90$ ($0.60$), the two-sided CI procedure yields
a two-sided observed CI when $d>2.448$ ($d>1.449$).
When $d\leq 0.320$, all three CI procedures result in empty intervals. 
In the figures, horizontal arrows indicate the area 
$A=\{\ d:\theta_{0}=1\in \text{CI}(d)\ \}$, 
where $\theta_{0}$ is the true value of $\theta$. 
Thus, if $d\in A$, the observed CI contains 
the true parameter value $\theta_{0}=1$. 
Furthermore, $P(A)=P(\theta_{0}\in \text{CI}(D))=\alpha$ 
implies that these three frequentist CI procedures have the correct coverage probabilities.
Vertical arrows indicate the observed CIs at $d=1$, $2$, and $3$. 
For instance, if $d=2$, the observed CI $[0,3.451)$ could be a realization of 
either a 95\% or 90\% frequentist CI procedure. 
Similarly, if $d=1$, the observed CI $[0,2.287)$ could be a realization of 
a 60\%, 90\% or 95\% frequentist CI procedure. 
Consequently, given an observed CI, 
its coverage probability (confidence level) may not be uniquely determined. 
This raises questions about the meaning of confidence level for an observed CI.

\begin{figure}[tbp]
\centering
\includegraphics[width=0.9\linewidth]{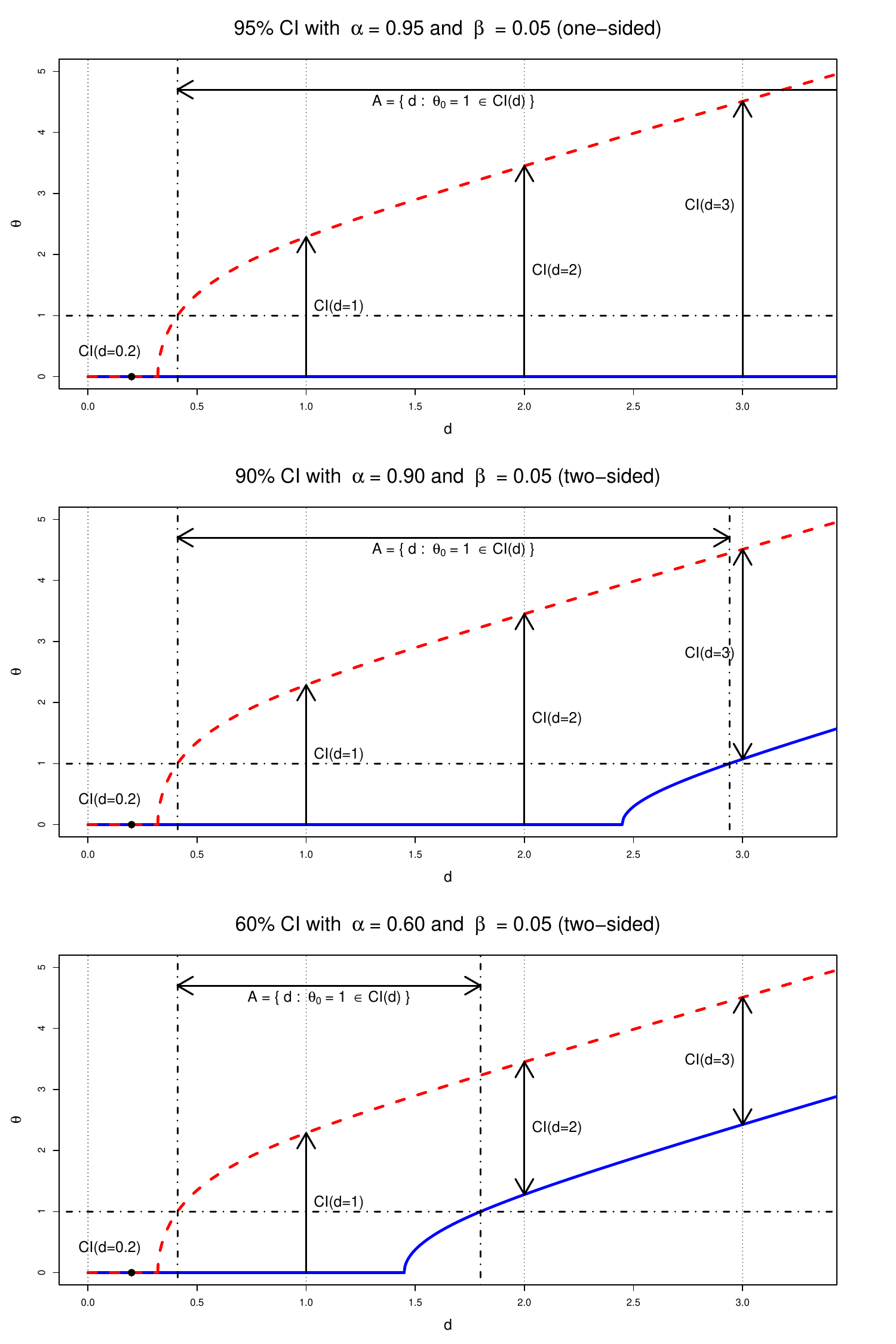}
\caption{Confidence intervals with 
(a) $\alpha =0.95$, (b) $\alpha =0.90$, (c) $\alpha =0.60$. 
For all three CIs, $\beta=0.05$.}
\label{fig:ci}
\end{figure}

\section{GFD and Probability Dilution}

\citet{fisher30} used a sufficient statistic to construct the CD. 
Since $(y_{1}, y_{2})$ are sufficient statistics for $(\mu_{1},\mu_{2})$, 
consider a joint CD for $(\mu_{1},\mu_{2})$ as follows: 
\[
C_{f}(\mu_{1},\mu_{2};y_{1},y_{2})
=P_{\mu_{1},\mu_{2}}(Y_{1}\geq y_{1}
\text{ and}Y_{2}\geq y_{2}). 
\]
This joint CD leads to a joint confidence density, 
\[
c_{f}(\mu_{1},\mu_{2};y_{1},y_{2})
=\frac{\partial^{2}C_{f}(\mu_{1},\mu_{2};y_{1},y_{2})}
{\partial \mu_{1}\partial \mu_{2}}
=\frac{1}{\sigma^{2}}\phi \left( \frac{\mu_{1}-y_{1}}{\sigma}\right) \phi 
\left( \frac{\mu_{2}-y_{2}}{\sigma}\right) 
=L(\mu_{1},\mu_{2};y_{1},y_{2}), 
\]
where $\phi (\cdot )$ is the density function of 
$N(0,1)$ and $L(\mu_{1},\mu_{2};y_{1},y_{2})$ is the likelihood. 
Thus, it is also a joint posterior 
under the UP for $(\mu_{1},\mu_{2})$. 
Integrated CD and marginal posterior of $\theta$ are identically obtained by,
\begin{equation}
G(\theta ;d)
=\int_{\mu_{1}^{2}+\mu_{2}^{2}\leq \theta^{2}}
c_{f}(\mu_{1},\mu_{2};y_{1},y_{2}) \ d(\mu_{1},\mu_{2})
=\Gamma_{2}\left( \frac{\theta^{2}}{\sigma^{2}};
\frac{d^{2}}{\sigma^{2}}\right),
\label{eq:up}
\end{equation}
which gives the density 
\begin{equation}
g(\theta ;d)
=\frac{\partial}{\partial \theta} G(\theta ;d)
=\frac{2\theta}{\sigma^{2}} \ \gamma_{2}\left( 
\frac{\theta^{2}}{\sigma^{2}};\frac{d^{2}}{\sigma^{2}}\right),
\end{equation}
where $\gamma_{2}(\theta;\cdot)=d\Gamma_{2}(\theta;\cdot)/d\theta$.
we denote them as UP.
Stein (1959) noted probability dilution of UP when $k\rightarrow \infty$. 
\citet{wilkinson77} and \citet{pedersen78}
showed that the UP cannot maintain the correct coverage probability.
All these problems remain to hold in satellite conjunction problem with $k=2$,
as we shall discuss.

\citet{hannig09} introduced the GFD as a generalization of the FD. 
We first show that the integrated CD can be viewed as a GFD. 
Consider a data generating mechanism 
\[
(Y_{1},Y_{2})=(\theta \cos \theta +\sigma U_{1},\theta \sin \theta +\sigma U_{2}), 
\]
where $U_{1}$ and $U_{2}$ are independent random variables from $N(0,1)$. 
We can define the set-valued function as 
\[
Q(y_{1},y_{2},U_{1}^{*},U_{2}^{*})
=\left\{ \ \theta :(y_{1},y_{2})=(\theta \cos \theta +\sigma U_{1}^{*},
\theta \sin \theta +\sigma U_{2}^{*})\ \right\} 
=\sigma \sqrt{\left( \frac{y_{1}}{\sigma} -U_{1}^{*}\right)^{2}
+\left( \frac{y_{2}}{\sigma}-U_{2}^{*}\right)^{2}}
\sim \sigma \sqrt{\chi _{2}^{2}\left(\frac{d^{2}}{\sigma^{2}}\right)}, 
\]
then $G(\theta ;d)$ satisfies the definition of GFD.
It is worth noting that the marginal CD and RP in the next section are also 
included in the class of GFD.
However, as the UP (integrated CD) and RP cannot maintain the correct coverage probability,
namely the confidence feature,
GFD may not be a desirable generalization of the CD.

\section{CD and Related Methods}

To resolve the probability dilution, 
\citet{wilkinson77} and \citet{cunen20} 
proposed the use of the marginally defined CD below.
In the context of Stein's and satellite conjunction problems,
we call this marginal CD the CD
because it maintains the confidence feature
as we shall show.
Let $(\theta,\psi )$ and $(D,T)$ be the polar coordinate representations 
of $(\mu_{1},\mu_{2})$ and $(Y_{1},Y_{2})$,
\[
(\mu_{1},\mu_{2})=(\theta \cos \psi ,\theta \sin \psi )
\quad \text{and} \quad 
(Y_{1},Y_{2})=(D\cos T,D\sin T), 
\]
where $\psi =\tan^{-1}(\mu_{2}/\mu_{1})$ and $T=\tan^{-1}(Y_{2}/Y_{1})$.
Here, the distributions of $T$ and $T|D$ still depend on 
both $\theta $ and $\psi$; hence, 
$D$ alone is not a sufficient statistic for $\theta $ under the full data $(y_{1},y_{2})$. 
If we have a maximal ancillary statistic we may derive the CD with full data information 
based on the conditional distribution of $D|A$ \citep{pawitan23}. 
But, in Stein's and satellite conjunction problems, 
the maximal ancillary statistic is not known. 
However, the current definition of the CD \citep{schweder16} 
only requires that 
\begin{equation}
C(\theta_{0};D)\sim \text{Uniform}[0,1]  
\label{eq:def}
\end{equation}
at the true value $\theta_{0}$ of $\theta $, which guarantees that the CD
maintains the confidence feature.
\citet{cunen20} derived the marginal CD for $\theta $, 
based on the statistics $D$ 
\begin{equation}
C(\theta ;d)=P_{\theta}(D\geq d)
=1-\Gamma_{2}\left( \frac{d^{2}}{\sigma^{2}};
\frac{\theta^{2}}{\sigma^{2}} \right),
\label{eq:cd}
\end{equation}
where $\Gamma_{2}(\cdot;\nu)$ denotes the non-central Chi-square distribution function, 
and they showed that this CD does not have probability dilution. 
This is coincide with Wilkinson's \citeyearpar{wilkinson77} 
marginal CD for Stein's problem with $k>2$. 
This CD satisfies the current definition (\ref{eq:def}) of the CD by maintaing the confidence feature.
Let $\Theta $ be the parameter space of $\theta$ 
and $\Omega_{D}$ is the sample space of $D$. 
It is worth noting that \citet{wilkinson77}
interpreted a point mass at $\theta=0$ as an unassigned probability,
since he considered the parameter space as $\Theta =(0,\infty)$.
Thus, in his view, the CD is not a probability to have $C(\Theta)<1$.
In this context, the point mass at zero looks paradoxical \citep{schweder16}.
However, if we assume $\Theta =[0,\infty )$, then $C(\Theta)=1$ holds without an unassigned probability.
\citet{martin21} insisted that the CD is also at risk of false confidence.
However, this arises from a misunderstanding about the CD, 
because zero is not accounted for within the parameter space, $\Theta=(0,\infty)$.
In this paper, we show that the above mentioned properties are not drawbacks,
but rather they are indeed advantages of the CD
to maintain the confidence feature
for all $\theta \in \Theta = [0,\infty)$.

With a slight abuse of notation, 
we denote the confidence assigned to 
a proposition $A\subset \Theta$ as
\[
C(A)=C(A;d)=C(\theta \in A)=\int_{A}c(\theta ;d)d\theta.
\]
The CD can have a point mass at $\theta =0$, since 
\[
C(\{0\})=1-\Gamma_{2}\left( \frac{d^{2}}{\sigma^{2}};0\right)
\geq 0.
\]
This CD maintains the confidence feature, 
\[
C(\theta_{0};D)=C([0,\theta_{0}];D)
=1-\Gamma_{2}\left( \frac{D^{2}}{\sigma^{2}};
\frac{\theta_{0}^{2}}{\sigma^{2}}\right) 
\sim \text{Uniform}[0,1], 
\]
to give correct coverage probability for any true value 
$\theta_{0}\in \Theta$.

Let $M(D)=C(\{0\};D)$ denote the point mass at $\theta =0$ and $M(d)$ denote
its realized value of the point mass. As $\sigma \rightarrow 0$ or $\theta
\rightarrow \infty $, the point mass $M(D)$ vanishes 
\[
M(D)\stackrel{p}{\rightarrow}1-\Gamma_{2}(\infty ;0)=0. 
\]
As $\sigma \rightarrow \infty $ or $\theta \rightarrow 0$, 
\[
M(D)\stackrel{d}{\rightarrow}\text{Uniform}[0,1]. 
\]
Here the confidence density can be expressed as 
\[
c(\theta ;d)=M(d)\cdot \mathbb{d}(\theta )+c_{+}(\theta ;d), 
\]
where $\mathbb{d}(\theta )$ denotes the Dirac delta function to give a point
mass at $\theta =0$ and $c_{+}(\theta ;d)=\partial P_{\theta}(D\geq
d)/\partial \theta$ for $\theta >0$.

\subsection{CD and Confidence Level of an Observed CI}

A point mass at a boundary prevents the probability dilution to maintain the
confidence feature. We first investigate a necessary and sufficient
condition for a point mass in the CD. Suppose that 
$\theta \in \Theta $ is
the parameter of continuous scalar statistic $D\in \Omega_{D}$ 
and the $1-\alpha$ quantile $q_{\alpha}(\theta)$ 
is a strictly increasing function of $\theta$ 
for any $\alpha \in (0,1)$. 
Then, we have the following theorem with the proof in Appendix. 
\begin{theorem}
Let $\partial \Omega_D$ and $\partial \Theta$ denote 
the boundary of $\Omega_D$ and $\Theta$, respectively.
Then, $C(\theta; d)$
has no point mass if and only if
\begin{equation}
\label{eq:mass}
q_\alpha(\theta) \to \partial \Omega_D
\quad \text{as }
\theta \to \partial \Theta,
\quad \text{for any} \alpha  \in (0, 1).
\end{equation}
\end{theorem}

\noindent 
\citet{pawitan23} considered a curved exponential model. 
Let $y=1$ be an observation from $Y\sim N(\theta ,\theta^{2})$ 
for $\theta \geq 0$, then one may consider a confidence distribution, 
\[
C(\theta ;y)=P_{\theta}(Y\geq y)
=1-\Phi \left( \frac{y-\theta}{\theta}\right) ,
\]
where $\Phi (\cdot )$ denotes the cumulative function of $N(0,1)$. 
However, this leads to 
\[
\lim_{\theta \rightarrow \infty}C(\theta ;y)=1-\Phi (-1)\approx 0.84<1.
\]
Here $C(\{0\};y=1)=0.$\ Thus, there is no point mass at $\theta =0.$
Following \citet{wilkinson77}, one may say that
this CD has an unassigned probability $0.16=1-0.84$. 
This problem occurs because the quantile function $q_{\alpha}(\theta)$ 
is not increasing function of $\theta$. 
Now let $d=|y|$ be an observation of $D=|Y|$ with $\Omega_{D}=\Theta$. 
Then the corresponding CD is defined as 
\[
C(\theta ;d)
=P_{\theta}(D\geq d)
=1-\Phi \left( \frac{d-\theta}{\theta}\right) 
+\Phi \left( \frac{-d-\theta}{\theta}\right),
\]
which becomes a proper distribution function without a point mass 
\[
\lim_{\theta \rightarrow 0}C(\theta ;d)
=1-\Phi (\infty )+\Phi (-\infty)=0
\quad \text{and}\quad 
\lim_{\theta \rightarrow \infty}C(\theta;d)
=1-\Phi(-1)+\Phi(-1)=1.
\]
When there is no point mass, under appropriate conditions, 
\citet{pawitan23} showed that 
\[
C(\theta \in CI(d))
=\int_{CI(d)}c(\theta ;d)d\theta 
=P_{\theta}(\theta \in CI(D)),
\]
where the LHS is the epistemic confidence of the observed CI 
and the RHS is the frequentist coverage probability of the CI procedure. 
Thus, the confidence of the
observed CI corresponds to the coverage probability of the CI procedure. 

In Stein's and satellite conjunction problems, 
lower bounds of $\Omega_{D}$ and $\Theta $ are zero. 
However, $q_{\alpha}(0)\neq 0 \in \partial \Omega_D$. '
Thus, the theorem implies that the corresponding CD has a point mass at zero. 
Now we extend relationship between frequentist coverage probability 
and epistemic confidence when there is a point mass. 
The confidence of the observed $\text{CI}_{\alpha}(d)$ is as follows, 
corresponding to the CI procedure \eqref{eq:CI} 
in Section \ref{s: freq}:

\begin{itemize}
\item[(a)]  When the observed CD is two-sided ($\theta_{L}(d)>0)$, 
we have 
\[
C(\text{CI}_{\alpha}(d))
=C(\theta <\theta_{U}(d);d)-C(\theta <\theta_{L}(d);d)
=(1-\beta )-(1-\alpha -\beta )
=\alpha .
\]

\item[(b)]  When the observed CD becomes one-sided 
($\theta_{L}(d)=0$ and $\theta_{U}(d)>0$), we have 
\[
C(\text{CI}_{\alpha}(d))
=C(\theta <\theta_{U}(d);d)
=1-\beta 
=\max \left\{
\ \alpha :\text{CI}_{\alpha}(d)=[0,\theta_{U}(d))\text{ for some}0\leq
\beta \leq 1-\alpha \ \right\} ,
\]
which is the maximum coverage probability (confidence level) among CI
procedures having the observed CI $[0,\theta_{U}(d))$. 
For example, in Figure \ref{fig:ci}, 
both $\alpha =0.95$ and $\alpha =0.9$ lead to the same observed CI 
$[0,3.451)$ for $d=2$. Here, the CD gives a confidence 
\[
C([0,3.451);d=2)=1-\beta =0.95.
\]

\item[(c)]  When the observed CI becomes an empty set 
($\theta_{L}(d)=\theta_{U}(d)=0$), we have 
\[
M(d)=C(\{0\};d)
=1-\Gamma_{2}\left( \frac{d^{2}}{\sigma^{2}};0\right) 
=\max \left\{ \ \alpha :\text{CI}_{\alpha}(d)
=\emptyset \ \text{ for some} 0 \leq \beta \leq 1-\alpha \ \right\},
\]
which is the maximum confidence level among CI procedures, having an empty interval. 
The point mass leads to a nice interpretation. 
In Figure \ref{fig:ci}, all the three procedures lead to 
$\text{CI}_{\alpha}(d=0.2)=\emptyset $. 
Here, the point mass of CD is 
\[
M(d=0.2)=C(\{0\};d)=0.980,
\]
which implies that the CI procedure produces an empty observed CI if $\alpha <0.98$. 
Thus, given $d=0.2$, we can allow the CI only with $\alpha \geq 0.98$.
Here the point $\{0\}$ has the confidence level 0.98.
\[
C(\theta \in CI(d))=\max P_{\theta}(\theta \in CI(D))
=\max \ \{\ \alpha :\text{CI}_{\alpha}(d)=CI(d) \
\text{ for some}0\leq \beta \leq 1-\alpha \ \}.
\]
\end{itemize}

\noindent Given an observed CI, its CD gives the maximum attainable coverage
probability (confidence level) among CI\ procedures,
\[
C(\theta \in CI(d))
=\max P_{\theta}(\theta \in CI(D))
=\max \ \{\ \alpha : \text{CI}_{\alpha}(d)=CI(d) \ 
\text{ for some}0\leq \beta \leq 1-\alpha \ \}.
\]

\subsection{CD versus UP}

\begin{figure}[tbp]
\centering
\includegraphics[width=0.9\linewidth]{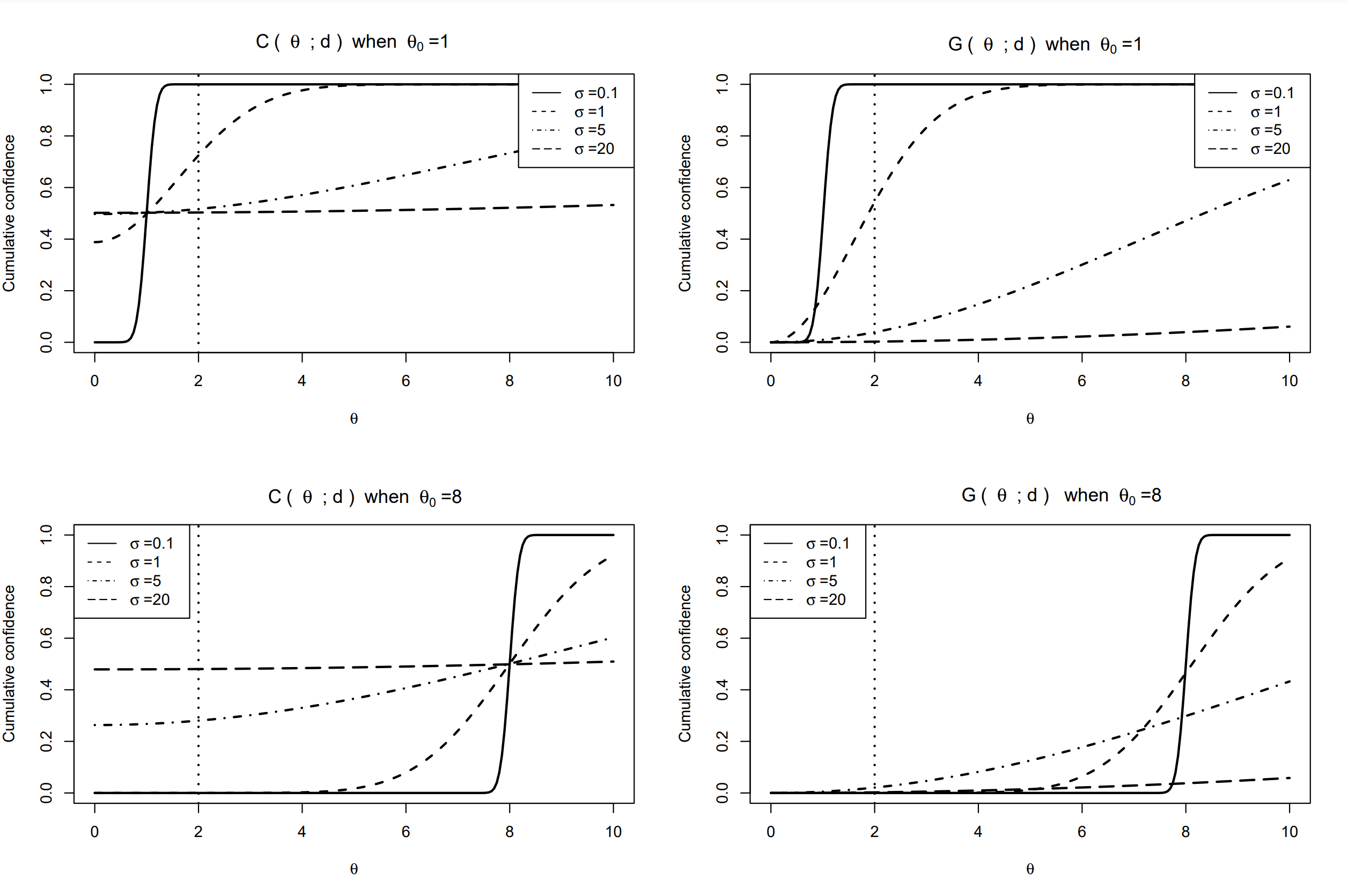}
\caption{Average of $C(\theta ;d)$ and $G(\theta ,d)$ over
10,000 repeats.}
\label{cumulatives}
\end{figure}
\citet{martin21} claimed that a drawback of CD is 
that it cannot have the two-sided CIs. 
We see that the CD allows a two-sided CI procedure, but it
can produce a one-sided observed CI to maintain the confidence feature.
We conjecture that any procedure
that always produces a two-sided interval
cannot maintain the confidence feature. 
Figure \ref{cumulatives} shows the averages of cumulative functions, based on CD \eqref{eq:cd} and UP \eqref{eq:up} from 10,000 repetitions, where $\theta_{0}$ is 1 or 8 and $\sigma$
varies in (0.1, 1, 5, 20). 
Both provide the cumulative distribution for $\theta $. 
Compared with the CD, the UP has apparent probability dilution.
The CD and UP become identical when $\sigma \rightarrow 0$. 
However, they can be quite different when $\sigma $ is large. 
Since the CD can have a point mass at zero, $C(\{0\})\geq0$, 
while the UP does not: $G(\{0\})=0$, 
the UP can always provide a two-sided interval, 
but it leads to probability dilution, 
losing the confidence feature as we shall see.

\subsection{CD versus RP}

\citet{bernardo79} proposed the RP to resolve 
the probability dilution of the UP when $\theta \ll k$. 
Figure \ref{coverage-two} shows the coverage probabilities 
of the one-sided and two-sided 80\% CIs for satellite conjunction problem $(k=2)$ 
and Stein's problem $(k=100)$, computed from 10,000 repetitions.
Probability dilution of UP is evident, especially when $\theta \ll k$. 
RP improves UP, but it also has probability dilution 
when $\theta$ is small or $\sigma$ is large.
Both RP and UP are the probability on $\Theta =(0,1)$
and always allow two-sided observed CIs,
because they do not have a point mass at zero.
This causes a probability dilution that 
both RP and UP cannot maintain the confidence feature.
Note that the CD-based two-sided CI procedure automatically produces the one-sided observed CI
to maintain confidence feature
when the observation $d$ is small. 
The figure shows that only the CD maintains the confidence feature 
for all $\theta \in \Theta = [0, \infty)$.

\begin{figure}[tbp]
\centering
\makebox{\includegraphics[width=\linewidth]{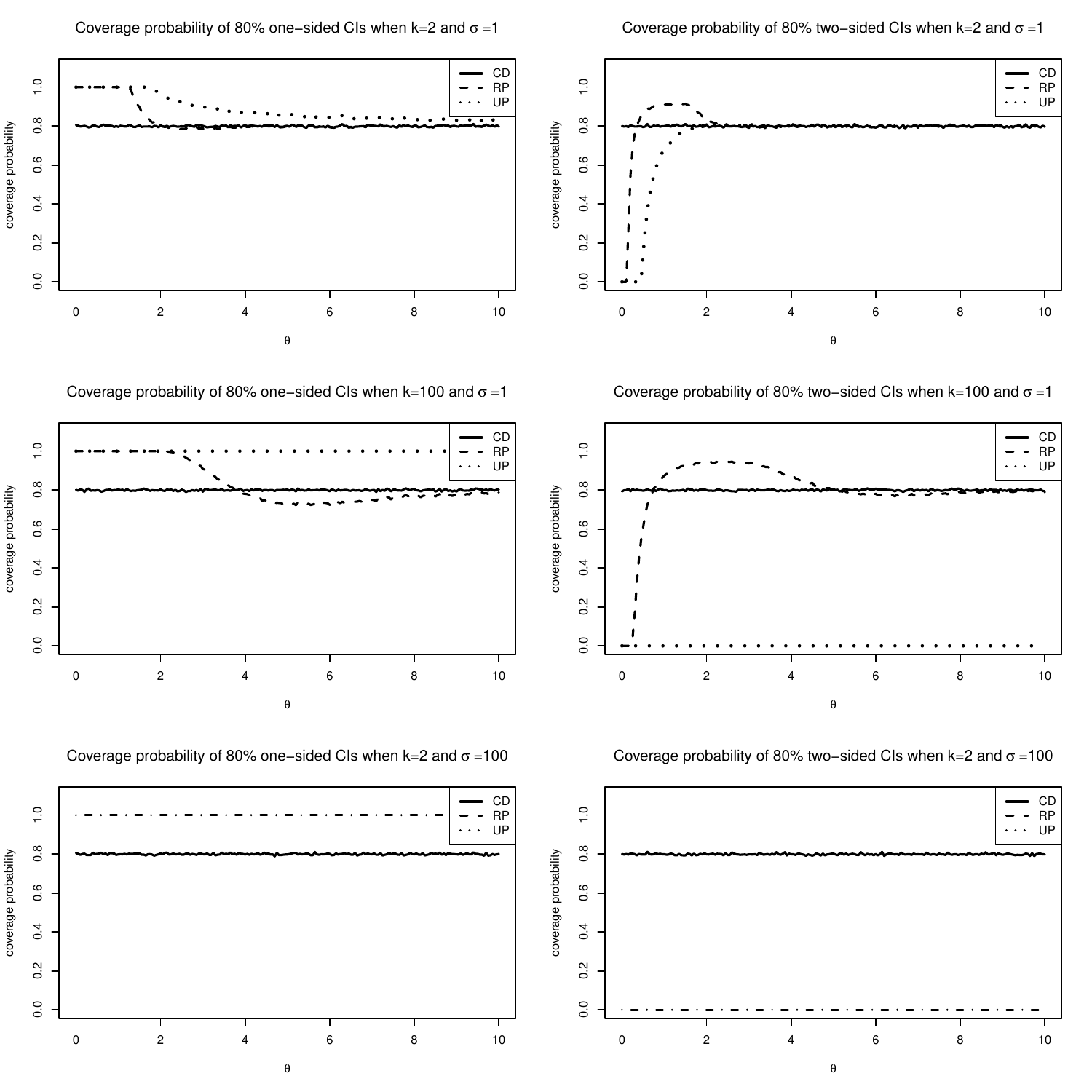}}
\caption{Coverage probabilities of 80\% CIs based on CD, UP, and RP 
when $k=2$ and $k=100$.}
\label{coverage-two}
\end{figure}

\section{On False Confidence}
\label{s: CB} 

\citet{balch19} noted that 
the inference about some proposition of interest $A$
suffers from severe false confidence if
for some unacceptably high $p\in (0,1)$
and unacceptably high $1-\alpha$ where $\alpha \in (0,1)$,
there exists some putative value of $\theta$ such that
$$
\theta \notin A 
\quad \text{and} \quad 
P_{\theta}\{C(A;D)\geq 1-\alpha \}\geq p.
$$
They claimed that this false confidence
is a fundamental deficiency in probabilistic inference,
and introduced the false confidence theorem below,
under the assumption that 
$\sup_{\theta} c(\theta;d) < \infty$ almost everywhere in $D$
for any true $\theta_{0}\in \Theta$.

\noindent \textbf{False confidence theorem \citep{balch19}} 
\textit{\ For any $\theta_0$, any $\alpha \in (0,1)$ and any $p\in (0,1)$,
there exists a subset $A\subset \Theta $ such that 
$\theta_{0}\notin A$ and}
\begin{equation}
P_{\theta_{0}}\{C(A;D)\geq 1-\alpha \}\geq p.  
\label{eq:false confidence}
\end{equation}

\noindent
In their response to \citet{cunen20},
\citet{mandel21} claimed that 
the false confidence theorem applies to 
not only the Bayesian posteriors, such as UP and RP,
but also the CD in satellite conjunction problem.
However, the false confidence theorem cannot be applied to the CD due to the presence of point mass.
For substantive understanding, 
consider a case of $\theta_{0}=0$,
which implies that $D$ follows the central Chi-square distribution.
Then, for any false proposition $A\subseteq \{\theta :\theta \neq 0\}$, 
\[
C(A;D)\leq C(\theta \neq 0;D)=1-M(D)\sim \text{Uniform}[0,1], 
\]
which leads to
\[
P_{\theta_{0}=0}\{C(A;D)\geq 1-\alpha \}\leq P_{\theta_{0}=0}\{M(D)\leq
\alpha \}=\alpha . 
\]
Thus, the false confidence theorem cannot be applied to the CD
in both Stein's and satellite conjunction problems.
\citet{balch19} proposed to use
the Martin-Liu validity criterion \citep{martin15}:
A statistical method is free from false confidence if
for any $\alpha \in [0,1]$ and any false proposition 
$A \subset \Theta$ such that $\theta_0 \notin A$,
\begin{equation}
\label{eq:martin-liu}
P_{\theta_{0}}\{C(A;D)\geq 1-\alpha \} \leq \alpha.     
\end{equation}
This criterion requires protection 
for any false proposition.
The existence of false proposition with a high confidence 
could be annoying if it is of interest.
However,
we may not need a protection from all false propositions.
For example, consider a false proposition
$A = \{ \ \theta: \theta \neq 1 \ \}$
with $\theta_0 = 1$.
Then, $A$ does not satisfy \eqref{eq:martin-liu}
for CD, UP, and RP, since $C(\{1\})=0$.
However, such a false proposition $A$ is not of interest 
since the true $\theta_0$ is unknown.
Thus, it may be enough to satisfy \eqref{eq:martin-liu}
only for the proposition of interest.
In satellite conjunction problem,
the proposition of interest is clearly identified,
whether the satellites collide or not.
Let $R$ be the sum of the radii of two satellites and let $H_{0}$ 
(collision; $\theta \leq R$) be the true, 
and let $H_{1}$ (non-collision; $\theta >R$) be the false proposition $A$. 
Then, the level of false confidence becomes 
\[
P_{\theta_{0}}\{C(H_{1};D)\geq 1-\alpha \}
=P_{\theta_{0}}\{C(H_{0};D)\leq \alpha \}
\leq P_{\theta_{0}}\{C(\theta_{0};D)\leq \alpha \}
=\alpha . 
\]
Hence, if $H_{0}$ is true, 
the level of false confidence cannot grow arbitrarily large. 
This satisfies the Martin-Liu validity criterion \eqref{eq:martin-liu} 
for $H_{1}$, the false proposition of interest. 
Thus, unlike the UP and RP,
the CD does not allow high false confidence for a false
proposition of interest.

To satisfy the Martin-Liu criterion,
\citet{martin21} proposed the use of CB,
\begin{equation}
\text{Bel}(A;d)=1-\sup_{\theta \in A^{c}}\text{pls}(\theta ;d),
\label{eq:CB}
\end{equation}
where $\text{Bel}(\cdot ;d)$ is the consonant belief function and 
$\text{pls}(\theta ;d)=1-|2C(\theta ;d)-1|$ is the plausibility contour. 
This CB can avoid the false confidence of any false proposition,
i.e., satisfy the Martin-Liu validity criterion.
However, this unnecessary thorough protection of CBs 
against all possible false propositions
make them vulnerable in other prospects of statistical inferences. 
Here we introduce the `null belief theorem',
which excludes the use of CBs,
from the opposite point of view of the false confidence theorem
under the same assumption. 
\begin{theorem}[Null belief theorem]
Consider a CB $\text{Bel}(\cdot ;d)$ characterized by either a CD or posterior probability.
Then, for any true $\theta_{0}\in \Theta $ and any $p\in (0,1)$,
there exists an interval $I$ with positive length such that
\begin{equation*}
\theta_{0}\in I\subset \Theta
\quad \text{and}\quad
P_{\theta_{0}}\{\text{Bel}(I;D)=0\}\geq p.
\end{equation*}
\end{theorem}
\begin{proof}
First, take a small interval near the true $\theta_{0}$. 
Let $\hat{\theta}(d)$ be the median of the CD such that 
$C\left( \hat{\theta}(d);d\right) =0.5$. 
Then, for any $p\in (0,1)$,\ there exists $\epsilon >0$ such that
\begin{equation*}
P_{\theta_{0}}
\{\hat{\theta}(D) \in (\theta_{0}-\epsilon ,\theta_{0}+\epsilon )\}
\leq 1-p.
\end{equation*}
Let $I=(\theta_{0}-\epsilon ,\theta_{0}+\epsilon )$ be an interval
that contains the true value $\theta_{0}$ but is such that
\begin{equation*}
P_{\theta_{0}}\{\text{Bel}(I;D)=0\}=P_{\theta_{0}}\{\hat{\theta}(D)\not\in
I\}\geq p,
\end{equation*}
then the theorem is proved. \hfill \qed
\end{proof}

Even when an interval $I$ contains the true value $\theta_{0}$, 
the probability of $\text{Bel}(I;D)=0$ is greater than arbitrary large $p>0$,
which seems counter-intuitive. 
On the other hand, the CD, along with the RP and UP, 
is free from the null belief theorem. 
While the belief function could be useful for the \textit{trinary} decision problem 
with the combination of an additional 
\textit{plausibility function} \citep{dempster68},
the null belief theorem implies that the belief function alone 
may not be suitable for a valid hypothesis testing, 
as we shall demonstrate later.

\subsection{False Confidence and Probability Dilution}

\label{s: dilution}

\citet{balch19} claimed that probability dilution 
is a symptom of the false confidence.
However, in this section, 
we emphasize the need to distinguish
probability dilution from false confidence. 
From now on, we will differentiate 
between confidence and probability, 
as confidence is an extended likelihood, 
not a probability \citep{pawitan21}.
Let $H_{0}:\theta \leq R$ be an assertion of collision and 
$H_{1}:$ not $H_{0}$ be an assertion of non-collision. 
Here, $G(H_{0})=G(\theta \leq R)=G([0,R])$ 
is the probability of collision based on the UP. 
Probability dilution means that 
\[
G(H_{0})\rightarrow 0\quad \text{as}\sigma \rightarrow \infty . 
\]
This causes a severe and persistent underestimate of risk exposure. 
Based on the CD, 
\citet{cunen20} investigated the confidence of collision $C(H_{0})$ 
at the true value $\theta_{0}=1.99\approx 2=R$, 
and $C(H_{0})=C([0,2])\approx C([0,1.99])\sim \text{Uniform}[0,1]$ 
for any $\sigma$. 
Thus, in their simulation study, 
the average confidence of collision remained close to $0.5$ for all $\sigma$. 
In general, the point mass is less than or equal to $C(\theta_{0};d)$, 
\[
M(d)=C(0;d)\leq C(\theta_{0};d)\sim \text{Uniform}[0,1], 
\]
but when $\sigma \rightarrow \infty $ or $\theta_{0}\rightarrow 0$, 
the point mass at $\theta =0$ converges in distribution to Uniform$[0,1]$, 
\[
M(D)=1-\Gamma_{2}\left( \frac{D^{2}}{\sigma^{2}};0\right) 
\stackrel{d}{\rightarrow}\text{Uniform}[0,1]. 
\]
Thus, as $\sigma \rightarrow \infty $, 
the average of $M(D)$ converges to 0.5.
This prevents the dilution of confidence.

\begin{figure}[tbp]
\centering
\includegraphics[width=0.9\linewidth]{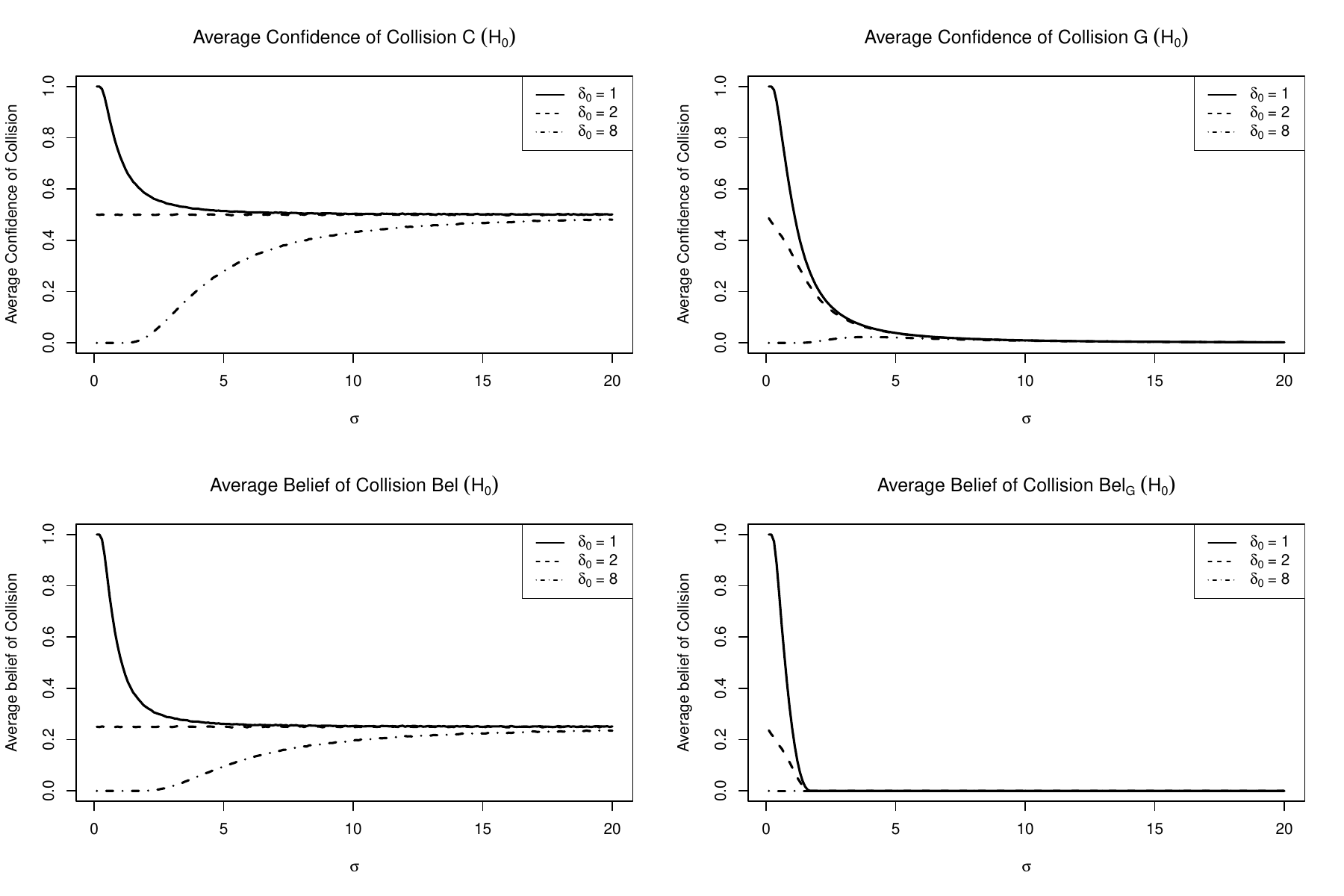}
\caption{Average confidences and beliefs of collision over 100,000
repetitions.}
\label{fig:collision}
\end{figure}

Figure \ref{fig:collision} shows the average confidences and beliefs of
collision as $\sigma $ varies from 0 to 20. As $\sigma $
increases, $C(H_{0})$ decreases to 0.5 when $\theta_{0}=1$, and $C(H_{0})$
increases to 0.5 when $\theta_{0}=8$. We see that these phenomena are
caused by a point mass at zero: $C(\{0\})>0$. 
$\text{Bel}(H_{0})$ is the CB \eqref{eq:CB} based on the CD, 
$C(\theta ;d)$. Bel$(H_{0})$ converges to 0.223 as $\sigma \rightarrow \infty $. 
We see that Bel$(H_{0})<C(H_{0})$,
i.e., the additional consonant feature in CB leads to dilution of confidence. 
The UP has no point mass, so $G(H_{0})$ goes to zero as 
$\sigma \rightarrow \infty $. 
$\text{Bel}_{G}(\cdot )$ is the CB function \eqref{eq:CB} based on the UP. 
$\text{Bel}_{G}([0,\theta ])$ has no point
mass, so it has a probability dilution. Thus, it is the
confidence feature, not the consonant feature, 
which prevents the dilution of confidence.

\section{Hypothesis Testing}

\label{s: test}
Hypothesis testing procedures are often used in risk assessment for
satellite conjunction problem. Depending on certain fundamental questions,
the null hypothesis $H_{0}$ can be either $\theta \leq R$ (collision) 
or $\theta > R$ (non-collision).
The probability (confidence) of collision is the most frequently used test
statistic in satellite conjunction problem \citep{hejduk19b}. For
illustration, we suppose that $H_{0}$ is the assertion of collision. 
Then, from the property of CD, 
\[
C(\theta \leq \theta_{0};D)\sim \text{Uniform}[0,1], 
\]
the confidence $C(H_{0};d)$ directly becomes the observed p-value for
testing $H_{0}$, i.e., 
\[
\max_{\theta \in H_{0}}P_{\theta}\left( C(H_{0};D)\leq \alpha \right)
=\alpha . 
\]
Thus, the CD yields $\alpha $-level hypothesis testing procedure for any $\sigma$. 
However, since probabilities such as the UP and RP have no point mass, 
as $\sigma \rightarrow \infty $ 
\[
G(H_{0};D)\rightarrow 0
\quad \text{and}\quad 
R(H_{0};D)\rightarrow 0. 
\]
Thus, if the data are of a poor quality ($\sigma $ is large), 
the UP and RP cannot accept the null hypothesis even if $d<R$. 
For example, suppose that we observe $d=1<R=2$. When $\sigma =1$, 
the CD, UP and RP give $C(H_{0})=0.918$, $G(H_{0})=0.731$ and $R(H_{0})=0.891$, respectively. 
Thus, all of them would not reject $H_{0}$. Here the RP becomes close to the CD.
However, in satellite conjunction problem, 
$\sigma$ is often much greater than $R$ \citep{balch19}. 
When $\sigma=100$, the CD yields $C(H_{0})=1.000$, 
hence the CD would not reject $H_{0}$. 
However, $G(H_{0})=0.000$ and $R(H_{0})=0.016 < 0.05$ to reject $H_{0}$ 
though the observed value implies the collision $(d<R)$. 
Therefore, if the collision risk is given by the CD, 
there is no reason for engineers to ignore an impending collision risk 
due to the negligible confidence of collision. 
However, if it is based on the UP or RP, 
engineers may ignore the impending danger 
because of the dilution of collision probability $G(H_{0})\approx 0$ 
and $R(H_{0})\approx 0$.
In Stein's problem it is often of interest to test 
\[
H_{0}:\ \theta =0\quad \text{vs.}\quad H_{1}:\ \theta \neq 0. 
\]
Due to the point mass at $\theta =0$, the CD gives 
\[
P_{\theta \in H_{0}}(C(H_{0};D)<\alpha )=P_{\theta =0}(C(\{0\};D)<\alpha
)=\alpha . 
\]
Thus, if we use $C(H_{0})$ as a p-value, we can directly achieve a valid
hypothesis testing procedure with 
\[
P_{H_{0}}(\text{Reject}H_{0})=P_{\theta =0}(C(H_{0})<\alpha )=P_{\theta
=0}(M(D)<\alpha )=\alpha . 
\]
On the other hand, the UP and RP have no point mass. 
Thus, $G(H_{0})$ and $R(H_{0})$ do not lead to a valid hypothesis testing, 
because $G(H_{0})=R(H_{0})=0$ for any observation $d$.

Now, suppose that we use the CB based on the CD in satellite conjunction problem, 
$\text{Bel}(H_{0})=\text{Bel}([0,R])$. Note here that the CB often becomes zero, 
due to the null belief theorem. When true $\theta_{0}=R$ (collision), we have 
\[
P_{\theta_{0}=R}\left\{ \text{Bel}(H_{0})=0\right\} 
=P_{\theta_{0}=R}\left\{ C(R;D)\leq 0.5\right\} =0.5. 
\]
Suppose that the CB is used for testing by rejecting $H_{0}$ 
if $\{$Bel$(H_{0})\leq \alpha \}$. Then, 
\[
P_{\theta_{0}=R}(\text{Reject}H_{0})
=P_{\theta_{0}=R}\left\{ \text{Bel}(H_{0})\leq \alpha \right\} 
\geq P_{\theta_{0}=R}\left\{ \text{Bel}(H_{0})=0\right\} =0.5. 
\]
Thus, though $H_{0}$ (collision) is true, the CB rejects $H_{0}$ with probability 0.5. 
It implies that the CB cannot achieve the significance level under 0.5. 
Similarly, Bel$_{G}(H_{0})$ cannot achieve the significance level under 0.847. 
Thus, the CBs cannot be applied to hypothesis testing for risk assessment.

When $\sigma =\infty $, the data $(y_{1},y_{2})$ are meaningless 
as an estimate of $(\mu_{1},\mu_{2})$. 
However, poor data cannot justify the small collision probability 
$G(H_{0})\approx 0$ under impending collision situations. 
The low collision probability ($G(H_{0})$ or $R(H_{0})$)
does not mean that the two satellites are far apart; 
it is only a statement of the general unlikelihood of such an alignment 
if all one knows is that the two satellites happen to be 
in the same general area \citep{hejduk19}. 
However, it is undesirable for engineers to ignore impending danger 
because they believe a negligible collision probability 
caused by poor data. 
Since $C(H_{0})=1-C(H_{1})\stackrel{d}{\rightarrow}\text{Uniform}[0,1]$ 
as $\sigma \rightarrow \infty$, 
the CD always acknowledges a non-negligible confidence of collision 
even with poor data. 
In this respect, the CD is useful for lowering the impending risk 
in satellite conjunction problem.

\section{Concluding Remarks}

The concept of Neyman's confidence provides 
an objective frequentist interpretation of the CI procedure, 
while the CD offers an epistemic interpretation of an observed CI. 
These two concepts are complementary,
as they are coincide under appropriate conditions.
The confidence of an observed CI can attain 
the objective frequentist interpretation 
of the coverage probability in repeated experiments.
Given the data, an observed CI from frequentist CI procedure 
can attain a meaningful epistemic interpretation via CD, 
especially when its coverage probability is ambiguous. 
We demonstrate in Figure \ref{fig:collision} 
that the presence of a consonant feature alone 
cannot prevent probability dilution
but rather accelerates it.
It is confidence feature that is the key to preventing 
a severe and persistent underestimate of collision probability. 
A point mass in the CD allows to preserve the confidence feature,
while probabilities, such as UP and RP,
lose the confidence feature.
GFD includes the UP (integrated CD) and RP as well as the marginal CD, 
so it is not an ideal generalization of CD 
that maintains the confidence feature.
Since the CD is an extended likelihood
\citep{pawitan21}, 
the integrated CD could no longer be a valid CD. 
The marginal CD \eqref{eq:cd} is based on the distribution of 
a non-sufficient statistic $D$, 
potentially not exploiting all the information in the data. 
If a maximal ancillary statistic exists, 
it is possible to derive the CD for $\theta$ 
by conditioning the maximal ancillary statistic 
without loss of information,
allowing no relevant subset \citep{pawitan23}.
It could be an interesting future research 
to find a way to construct a CD 
with the full data likelihood in general.
Many aspects of the CD still require further investigation and understanding.

\bibliographystyle{apalike}
\bibliography{references}

\section*{Appendix}

\subsection*{Proof of necessary and sufficient condition for existence of a point mass}

\label{s: iff}

\begin{proof}
Let $D_U$ and $D_L$ be the upper and lower bounds of $\Omega_D$,
and let $\theta_U$ and $\theta_L$ be the upper and lower bounds of $\Theta$.
Since the quantile function is continuous,
we write $q_\alpha(\theta_L) = \lim_{\theta\to\theta_L}q_\alpha(\theta)$
and $q_\alpha(\theta_U) = \lim_{\theta\to\theta_U}q_\alpha(\theta)$.
Note here that we allow the bounds to be $\pm \infty$.

\noindent ($\Rightarrow$)
Suppose that there exists $0 \leq \alpha \leq 1$ such that
$$
q_\alpha(\theta_L) \neq D_L
\quad \text{or} \quad
q_\alpha(\theta_U) \neq D_U.
$$
If $q_\alpha(\theta_L) \neq D_L$, 
there exists $d^*$ such that
$D_L < d^* < q_\alpha(\theta_L)$.
Since $C(\theta; d^*)$ does not have a point mass,
$$
C(\theta_L; d^*) = P_{\theta_L}(D\geq d^*) = 0.
$$
However, by definition of the quantile,
$$
C(\theta_L; d^*)
> C(\theta_L; q_\alpha(\theta_L))
= P_{\theta_L}(D\geq q_\alpha(\theta_L))
= \alpha > 0.
$$
This leads to contradiction.
If $q_\alpha(\theta_U) \neq D_U$, 
there exists $d^*$ such that
$q_\alpha(\theta_U) < d^* < D_U$.
Similarly,
$$
C(\theta_U; d^*) = P_{\theta_U}(D\geq d^*) = 1,
$$
but we have
$$
C(\theta_U; d^*)
< C(\theta_U; q_\alpha(\theta_U))
= P_{\theta_U}(D\geq q_\alpha(\theta_U))
= \alpha < 1,
$$
which leads to contradiction.
Thus, 
$q_\alpha(\theta) \to \partial \Omega_D$
as $\theta \to \partial \Theta$ 
for all $\alpha \in (0, 1)$.

\noindent ($\Leftarrow$) 
Let $d$ be an arbitrary value in $\Omega_D$,
then \eqref{eq:mass} leads to
$$
C(\theta_L; d) 
= P_{\theta_L}(D\geq d)
\geq P_{\theta_L}(D\geq q_\alpha(\theta_L))
= \alpha.
$$
Taking $\alpha=0$ leads to $C(\theta_L; d) = 0$.
Similarly we can obtain $C(\theta_U; d) = 1$.
Thus, $C(\theta;d)$ does not have a point mass for any $d\in \Omega_D$.
\end{proof}

\end{document}